\newif\ifshort

\documentclass[a4paper,USenglish,cleveref,autoref,thm-restate]{lipics-v2021}

\hideLIPIcs  

\newcommand{\remove}[1]{}

\newcommand{\defproblema}[3]{
  \vspace{3mm}
\noindent\fbox{
  \begin{minipage}{.95\textwidth}
  \begin{tabular*}{\textwidth}{@{\extracolsep{\fill}}lr} #1  \\ \end{tabular*}
  {\bf{Input:}} #2  \\
  {\bf{Question:}} #3
  \end{minipage}
  }
  \vspace{2mm}
  }

\newcommand{\cA}{{\mathcal A}}

\newcommand{\OO}{\mathcal{O}}
\newcommand{\PP}{\ensuremath{\mathcal{P}}\xspace}

\newcommand{\nn}{{\mathbb N}}

\newcommand{\SpanTree}{{\sc Spanning Tree}}
\newcommand{\MLST}{{\sc Max-Leaf Spanning Tree}\xspace}
\newcommand{\MIST}{{\sc Max-Internal Spanning Tree}\xspace}
\newcommand{\NTST}{{\sc Non-Terminal Spanning Tree}\xspace}

\newcommand{\DivMLST}{{\sc Diverse Max-Leaf Spanning Tree}}
\newcommand{\DivMIST}{{\sc Diverse Max-Internal Spanning Tree}}
\newcommand{\DivNTST}{{\sc Diverse Non-Terminal Spanning Tree}}

\newcommand{\DivLIST}{\textsc{Leaf\&Internal-Constrained Diverse Spanning Trees}\xspace}
\newcommand{\DivLNTST}{\textsc{Leaf\&Non-terminal-Constrained Diverse Spanning Trees}\xspace}

\newcommand{\dLIST}{\textsc{Leaf\&Internal-CDST}\xspace}
\newcommand{\dLNTST}{\textsc{Leaf\&Non-Term-CDST}\xspace}

\newcommand{\kDiv}{{pairwise $k$-diverse}}

\newcommand{\dg}{{\rm deg}}

\newcommand{\lv}[1]{}

\usepackage{tikz}
\usetikzlibrary{decorations.shapes,decorations.pathreplacing,decorations.pathmorphing}
\usetikzlibrary{arrows,matrix,shapes}
\usetikzlibrary{positioning}
\tikzset{
        stars/.style={star,inner sep=2pt}
    }

\usepackage{todonotes}
\usepackage[ruled,vlined,linesnumbered]{algorithm2e}
\usepackage{hyperref}
\usepackage{xcolor}
\usepackage{cleveref}

\usepackage{microtype}

\newtheorem{reduction rule}{\bf Reduction Rule}[section]

\newtheorem{reductionrule}{Reduction Rule}

\title{Polynomial Kernels for Spanning Tree with Diversity Requirements} 

\author{Petr A. Golovach}{University of Bergen, Norway}{petr.golovach@uib.no}{https://orcid.org/0000-0002-2619-2990}{Supported by the Research Council of Norway via
BWCA (grant 314528) and Extreme-Algorithms (grant 355137) projects.}

\author{Diptapriyo Majumdar}{Indraprastha Institute of Information Technology Delhi, New Delhi, India}{diptapriyo@iiitd.ac.in}{https://orcid.org/0000-0003-2677-4648}{}

\author{Saket Saurabh}{The Institute of Mathematical Sciences, HBNI, Chennai, India}{saket@imsc.res.in}{https://orcid.org/0000-0001-7847-6402}{Supported by European Research Council under the European Union's Horizon 2020 research and innovation program (grant agreement No. 819416)}

\authorrunning{P. A. Golovach,  D. Majumdar, S. Saurabh} 

\Copyright{P. Golovach, D. Majumdar, S. Saurabh} 



\category{} 

\relatedversion{} 

\nolinenumbers 

\EventEditors{Jan Goedgebeur and Pawe{\l} Rz\k{a}\.{z}ewski}
\EventNoEds{2}
\EventLongTitle{52nd International Workshop on Graph-Theoretic Concepts in Computer Science (WG 2026)}
\EventShortTitle{WG 2026}
\EventAcronym{WG}
\EventYear{2026}
\EventDate{June 2--4, 2026}
\EventLocation{Kortrijk, Belgium}
\EventLogo{}
\SeriesVolume{376}
\ArticleNo{30}

\begin{document}

\maketitle


\begin{abstract}
Given a connected undirected graph $G$, a spanning tree is a subgraph $T$ of $G$ such that $V(T) = V(G)$ and $T$ is a tree.
A collection of $\ell$ spanning trees $T_1,\ldots,T_{\ell}$ is {\kDiv} if for every $i \neq j$, $|E(T_i) \triangle E(T_j)| \geq k$.
Given a connected undirected graph $G$ and integers $p, q, k, \ell$, {\DivLIST} asks whether there are $\ell$ distinct spanning trees $T_1,\ldots,T_{\ell}$ of $G$ that are {\kDiv} such that each tree has at least $p$ leaves and at least $q$ internal vertices.
Similarly, {\DivLNTST} takes a connected undirected graph $G$, $V_{\rm NT}\subseteq V(G)$, and three integers $p, k, \ell$, and asks if $G$ has $\ell$ spanning trees that are {\kDiv}, and each has at least $p$ leaves and contains the vertices of  $V_{\rm NT}$ as internal.
We consider these two problems from the kernelization perspective and provide polynomial kernels for {\DivLIST} and {\DivLNTST}, when parameterized by $p + q + k + \ell$ and $p + |V_{\rm NT}| + k + \ell$, respectively.
\end{abstract}

\section{Introduction}
\label{sec:intro}
The \emph{diversity} paradigm in the Parameterized Complexity framework was formulated by Fellows and Rosamond~\cite{FernauGS18}, but the concept was known even before, e.g., in Integer Programming~\cite{DannaW09}.
The paradigm is motivated by the common situation that arises in various applications when it is not sufficient to find just one high-quality solution to a combinatorial optimization problem due to numerous side constraints.  Furthermore, it makes no sense to try to find all optimal solutions because typically there are too many --- it is either computationally infeasible to list all of them or, if the first task is somehow achieved, impossible to analyze the huge list of solutions. 
Instead, it is preferable to work with a reasonably sized sample of sufficiently good solutions, which may be analyzed by making use of side information.
Such a sample should represent the variety of solutions, and the natural way to achieve this is to demand that the solutions be ``far apart''. 
This ``far apart'' notion can be mathematically expressed in various ways.
However,  when dealing with subset problems, it is common to measure the diversity of a sample as either the minimum pairwise Hamming distance between solutions, or the sum of all pairwise Hamming distances, or the cardinality of the union of the solutions.
Arguably, the first approach is most popular, especially in the context of Parameterized Complexity, as the sample size and the minimum Hamming distance are natural parameters in this case. 
The diversity paradigm attracted a lot of attention in recent years (see \cite{BasteFJMOPR22,BasteJMPR19,BergMS23,DrabikMasarik24,FominGJPS24,FominGPPS24,HanakaK0KKO23,HanakaKKLO22,HanakaKKO21,MisraM024} and the bibliography therein). In our paper, we consider the problem of finding diverse spanning trees with additional constraints on leaves and internal vertices.

Spanning trees are fundamental in Graph Theory and Computer Science, and various types of spanning tree problems naturally arise in applications. One of these variants is the {\MLST} 
problem~\cite{BonsmaBW03,DowneyF99,Estivill-CastroFLR05,Jansen12,LuR98,Solis-ObaBL17}. Formally, {\MLST} asks, given a graph $G$ and an integer $p$, whether $G$ has a spanning tree with at least $p$ leaves. The \MIST problem~ \cite{BjorklundKKZ17,ChenHGW18,FominGST13,HanakaK25,KnauerS15,Li0CW17,LiZW21,PrietoS03,PrietoS05} is dual to \MLST, and given a graph $G$ and an integer $q$, asks whether $G$ has a spanning tree with at least $q$ internal vertices.  The input to \NTST~\cite{HanakaK25} contains a graph $G$ with a set of \emph{non-terminal} vertices $V_{\rm NT}\subseteq V(G)$, and the question is whether  $G$ contains a spanning tree with the vertices of $V_{\rm NT}$ being internal. 

We consider the ``diverse'' variants of these problems.  For an integer $k\geq 0$, two spanning trees $T_1$ and $T_2$ are said to be \emph{$k$-diverse} if $|E(T_1)\triangle E(T_2)|\geq k$.
Respectively, a collection of spanning trees $T_1,\ldots,T_{\ell}$ is {\em {\kDiv}} if for every $i \neq j$, $T_i$ and $T_j$ are $k$-diverse. Since our approach for the considered variants of  \MLST and \MIST is based on similar arguments, we combine the requirements for leaves and internal vertices in the \DivLIST
 (\dLIST) problem.

\defproblema{\dLIST}%
{A connected graph $G = (V, E)$, and integers $p\geq 0$, $q\geq 0$, $k\geq 1$, and $\ell\geq 1$.}%
{Decide whether $G$ has $\ell$ pairwise $k$-diverse spanning trees $T_1,\ldots,T_\ell$ such that every tree has at least $p$ leaves and $q$ internal vertices.
}

Similarly,  {\MLST} and {\NTST} are generalized in \DivLNTST
(\dLNTST).

\defproblema{\dLNTST}%
{A connected graph $G$, a subset of non-terminal vertices $V_{\rm NT} \subseteq V(G)$, and integers $p\geq 0$, $k\geq 1$, and $\ell\geq 1$.}%
{Decide whether $G$ has $\ell$ pairwise $k$-diverse spanning trees $T_1,\ldots,T_\ell$ such that the vertices of $V_{\rm NT}$ are internal in every tree and each tree has at least $p$ leaves.
}

\subparagraph*{Our results.} As \MLST is the dual of the \textsc{Connected Dominating Set} problem, it is NP-complete~\cite{GareyJ79}. Also, {\MIST} and \NTST generalize \textsc{Hamiltonian Path} and $s$-$t$-\textsc{Hamiltonian Path}, respectively, which are well-known to be NP-complete~\cite{GareyJ79}. On the positive side, these problems are well-studied from the perspective of parameterized algorithms.
In particular, {\MLST}~\cite{BonsmaBW03,Estivill-CastroFLR05,Jansen12} and  \MIST~\cite{FominGST13,Li0CW17,PrietoS03,PrietoS05}  are FPT when parameterized by $p$ and $q$, respectively, and, furthermore, admit polynomial kernels for these parameterizations. The same holds for  {\NTST}~\cite{HanakaK25} parameterized by the size of the given set of non-terminal vertices. Our first result demonstrates that the kernelization carries over to the common diverse variant. 

\begin{restatable}{theorem}{LIST}
\label{thm:kernLIST}
{\dLIST} parameterized by $p + q + k + \ell$ admits a kernel with $\OO((p+q+k\ell)\ell)$ vertices.
\end{restatable}

Similarly, \NTST admits a polynomial kernel for the parameterization by the size of the set of non-terminal vertices~\cite{HanakaK25}, and we get the extension stated in the following theorem.

\begin{restatable}{theorem}{LNPST}
\label{thm:kernLNTST}
{\dLNTST} parameterized by $|V_{\rm NT}| + p + k + \ell$ admits a kernel with  $\OO((|V_{\rm NT}|+p+k\ell)\ell)$  vertices.
\end{restatable}

The key structural observation behind the proofs of \Cref{thm:kernLIST} and \Cref{thm:kernLNTST} is the fact that if a connected graph $G$ of minimum degree at least two has a spanning tree with at least $2\ell \big\lceil\frac{k}{4}\big\rceil$ leaves, then $G$ has $\ell$ pairwise $k$-diverse spanning trees. Using this, we construct a polynomial-time algorithm that, given an instance of \dLIST, either reports that this is a yes-instance, or returns an equivalent instance of bounded size, or outputs an equivalent instance of \MIST. Thus, we show that, from the kernelization viewpoint, \dLIST is as hard as \MIST. Similarly, for \dLNTST, we obtain that kernelization for finding diverse spanning trees is essentially as hard as for finding one tree with a given set of non-terminals.

\subparagraph*{Related work.}
The study of diverse spanning trees was initiated by Hanaka et al.~\cite{HanakaKKO21}. However, they considered the different diversity measure -- the sum of all pairwise Hamming distances. They proved that $\ell$ spanning trees maximizing this measure can be found in polynomial time. Moreover, this result can be extended to the problem of finding diverse matroid bases; we remind that the edge sets of spanning trees are bases of graphic matroids. The result is obtained by reducing finding diverse bases to the problem of finding disjoint bases, which could be solved in polynomial time by the classical results of Nash-Williams~\cite{NashWilliams64} and Edmonds~\cite{Edmonds65}.  We note that this approach fails for our case when we maximize the minimum Hamming distance between trees. 
Diverse matroid bases when the diversity measure is the minimum Hamming distance were investigated by Fomin et al. in~\cite{FominGPPS24}. They proved that the problem of finding $\ell$ pairwise $k$-diverse bases admits a polynomial kernel for the parameterization by $\ell$ and $k$.  Deciding whether a given graph has $\ell$ pairwise $k$-diverse spanning trees is a special case of this problem. 
In fact, Fomin et al. proved their result for the more general weighted problem where the diversity is defined as the weight of the symmetric difference of two bases. 

The first FPT algorithm for  \MIST parameterized by the number of leaves $p$ was given by Downey and Fellows~\cite{DowneyF99}, and a faster algorithm was provided by Bonsma et al.~\cite{BonsmaBW03}. The first polynomial kernel was given by Estivill-Castro et al.~\cite{Estivill-CastroFLR05}. The kernel with at most $5.5p$ vertices was established by Jansen~\cite{Jansen12}. Approximation algorithms for the problem were given in 
\cite{LuR98,Solis-ObaBL17}. As we are interested mainly in the parameterized complexity for the parameterization by the number of leaves, we do not discuss here the results for the \textsc{Connected Dominating Set} problem, where the task is to minimize the number of internal vertices of a spanning tree.  

The Parameterized Complexity study of \MIST was initiated by Prieto-Rodriguez and Sloper in~\cite{PrietoS03,PrietoS05} who gave a first FPT algorithm for the parameterization by the number of internal vertices $q$ and proved that the problem admits a polynomial kernel.
The currently best FPT algorithm was given by Bj{\"{o}}rklund et al.~\cite{BjorklundKKZ17}.
Fomin et al.~\cite{FominGST13} demonstrated that \MIST admits a kernel with at most $3q$ vertices, and Li et al.~\cite{Li0CW17} got the improved kernel with at most $2q$ vertices.
{\NTST} was considered by Hanaka and Kobayashi~\cite{HanakaK25}. In particular, they proved that the problem has a kernel with at most $3k$ vertices and gave a single-exponential in $k$ algorithm when the problem is parameterized by the size $k$ of the set of non-terminals.
Approximability of {\MIST} was investigated in~\cite{ChenHGW18,KnauerS15,Li0CW17,LiZW21}.  

Finally, we briefly survey some related results about diversity. The systematic study of the diversity problems in the Parameterized Complexity framework was initiated by Baste et al.~\cite{BasteFJMOPR22}.  
For the diverse variants of the classical subset optimization problems, we mention that problems of finding diverse matchings were considered in~\cite{FominGJPS24,FominGPPS24}, diverse paths in~\cite{HanakaKKLO22,HanakaKKO21}, diverse $s$-$t$-cuts in \cite{BergMS23}, diverse hitting sets in~\cite{BasteJMPR19}, and diverse satisfying assignments for \textsc{Satisfiability} in~\cite{MisraM024}.  
Solving the diversity problems under structural parameterizations was discussed in~\cite{BasteFJMOPR22,DrabikMasarik24}. Approximation algorithms for finding diverse solutions were considered in~\cite{HanakaK0KKO23}. We refer to \cite{DannaW09} for diversity for the (mixed) integer programming problems. 

\subparagraph*{Organization of the paper.} The rest of the paper is organized as follows. In \Cref{sec:prelims}, we introduce basic notions used throughout the paper. In \Cref{sec:technical}, we provide auxiliary technical results, and in \Cref{sec:kern}, we prove  \Cref{thm:kernLIST} and \Cref{thm:kernLNTST}. We conclude in  \Cref{sec:conclusion} by stating some open problems.

\section{Preliminaries}
\label{sec:prelims}

\subparagraph*{Sets, Numbers, and Graph Theory:}
We use $\nn$ to denote the set of all natural numbers.
Given an integer $r\geq 1$, the set $[r]$ denotes the set $\{1,\ldots,r\}$.
Given two sets $A$ and $B$, the {\em symmetric difference} between $A$ and $B$ is denoted by $A \Delta B$ which equals to the set $(A \setminus B) \cup (B \setminus A)$.
We use standard graph theoretic terminologies from Diestel's book \cite{Diestel-Book}.
For a graph $G = (V, E)$, let $V(G) = V$ denote the set of vertices and $E(G) = E$ denote the set of edges of $G$.
If the graph is clear from the context, we use $n$ to denote the number of vertices, i.e. $|V(G)|$, and $m$ to denote the number of edges, i.e. $|E(G)|$.
For a set $X \subseteq V(G)$, we use $G[X]$ to denote the subgraph induced by the vertex subset $X$.
Given a vertex $v \in V(G)$, let ${\dg}_G(v)$ denote the degree of $v$ in $G$, i.e. the number of neighbors of $v$ in $G$.
A {\em pendant vertex} is a vertex of degree exactly one.
A {\em pendant edge} is the unique edge that is incident to a pendant vertex. 
A \emph{tree} is a connected graph without cycles.
Given a graph $G$, a \emph{spanning tree} $T$ of $G$ is a tree subgraph of $G$ whose set of vertices is $V(G)$.
For a tree $T$, pendant vertices are \emph{leaves} and the other vertices are \emph{internal}.
A {\em path} $P$ in a graph is a sequence of distinct vertices $(v_0,\ldots,v_{\ell})$ such that for every $i \in [\ell]$, $v_{i-1}v_{i}$ is an edge; $\ell$ is the {\em length} of $P$.
The vertices $v_0$ and $v_\ell$ are the \emph{end-vertices}, and the other vertices are \emph{internal}.  
A {\em degree-2-path}  $P = (v_1,\ldots,v_{\ell})$  is a sequence of distinct vertices with only possible exception $v_0=v_\ell$ such that $v_1,\ldots,v_{\ell-1}$ are vertices of degree two in $G$; in the same way as for standard paths, $v_0$ and $v_\ell$ are end-vertices and the other vertices are internal, and $\ell$ is the length.

\subparagraph*{Parameterized Complexity and Kernelization:}
A {\em parameterized problem} $L$ is a set of instances $(x, k)$ where $x \in \Sigma^*$ over a finite alphabet $\Sigma$ and $k \in \nn$ is the parameter.
A parameterized problem $L$ is said to be {\em fixed-parameter tractable} if given $(x, k) \in \Sigma^* \times \nn$, there exists an algorithm $\cA$ that runs in $f(k)|x|^{\OO(1)}$-time, where $f$ is a computable function, and correctly decides whether $(x, k) \in L$ or not.
We call this algorithm $\cA$ a {\em fixed-parameter algorithm} (or {\em FPT algorithm} in short).
Note that we allow combinatorial explosion in the parameter $k$ while the algorithm runs in time polynomial in $|x|$.
We say that two instances $(x, k)$ and $(x', k')$ of a parameterized problem $L$ are {\em equivalent} when $(x, k) \in L$ if and only if $(x', k') \in L$.
A {\em kernelization} for a parameterized problem $L \subseteq \Sigma^* \times \nn$ is an algorithm that runs in time polynomial in $|x| + k$ and outputs $(x', k')$ such that $|x'| + k' \leq g(k)$ for some computable function $g: \nn \rightarrow \nn$.
If $g(k)$ is $k^{\OO(1)}$, then $L$ is said to admit a {\em polynomial kernel}.
A kernelization algorithm usually consists of a collection of {\em reduction rules} (or preprocessing rules) that have to be applied exhaustively.
A reduction rule is {\em safe} if given $(x, k)$ of $L$, one application of the reduction rule outputs an equivalent instance $(x', k')$ of $L$.
We refer to \cite{CyganFKLMPPS15,DowneyF13,FominLSZ19} for the formal introduction to the parameterized complexity theory and kernelization.

\section{Technical lemmata}\label{sec:technical}

In this section, we provide auxiliary statements  
that will be used in our kernelization algorithms. We start with the following folklore observation.

\begin{observation}\label{obs:leavesvsint}
Let $T$ be a tree with at least two vertices, and let $V_1$ and $V_{\geq 3}$ be the sets of vertices of degree one and at least three, respectively. Then $|V_{\geq 3}|\leq |V_1|-2$.  
\end{observation}

We are using as black boxes  the currently best kernelization results for \MIST and \NTST given by Li et al.~\cite{Li0CW17} and Hanaka and Kobayashi~\cite{HanakaK25}, respectively.
Recall that the task of \MIST is, given a graph $G$ and an integer $q$, to decide whether $G$ has a spanning tree with at least $q$ internal vertices.

\begin{proposition}[Theorem 1.1 of \cite{Li0CW17}]
\label{prop:MIST-dec-result}
{\MIST} admits a kernel with $2q$ vertices.
\end{proposition}

We remind that \NTST asks whether a graph $G$ given together with a subset of vertices $V_{\rm NT}$ has a spanning tree where the vertices of $V_{\rm NT}$ are internal.

\begin{proposition}[Theorem 3.1 of \cite{HanakaK25}]
\label{prop:NTST-dec-result}
{\NTST} admits a kernel with $3|V_{\rm NT}|$ vertices.
\end{proposition}

\ifshort
\else
The proof of the following lemma is implicit in~\cite{Estivill-CastroFLR05} (see also~\cite{DowneyF13}); we provide its proof for completeness. 
\fi
Let $P = (v_0,\ldots,v_{\ell})$ be a degree-2-path in $G$ such that $\ell \geq 6$.
We say that the vertices $v_3,\ldots,v_{\ell - 3}$ are {\em strictly internal} vertices of $P$.

\begin{lemma}
\label{lemma:non-tree-edge-increase-leaves}
Let $T$ be a spanning tree of $G$ that has $s$ leaves, and let $P$ be a degree-2-path in $T$ of length at least $6$.
Furthermore, assume that a strictly internal vertex $v$ of $P$ is adjacent to $w$ in $G$ such that $vw$ is not an edge of $T$.
Then, there exists a spanning tree $\widehat{T}$ of $G$ that has at least $s + 1$ leaves. 
Moreover, for any leaf $z$ of $\widehat{T}$ that is not a leaf of $T$,  $z$ is an internal vertex of $P$.
\end{lemma}

\ifshort
\else
\begin{proof}
Consider a spanning tree $T$ of $G$ with $s$ leaves.
Moreover, suppose that $P=(x_1,\ldots,x_d)$ is a degree-2-path of $T$ with $d\geq 7$, and a strictly internal vertex $v \in P$ is adjacent to a vertex $w$ such that $vw \notin E(T)$.
We assume that $T$ is rooted in $x_1$ and prove the statement by considering the following list of mutually exhaustive cases.
\begin{description}
	\item[Case (i):] $w$ is neither an ancestor nor a descendent of $v$ in $T$. Then $x_1$ is the least common ancestor of $v$ and $w$ in $T$.
	Consider the unique path $P_{x_1, w}$ between $x_1$ and $w$ in $T$ and the unique path $P_{x_1, v}$ between $x_1$ and $v$ in $T$.
	As $v$ is a strictly internal vertex of the degree-2-path $P$,  $P_{x_1, v}$ contains the degree-2-path between $x_1$ and $v$ such that $v$ is a vertex in the subpath $(x_4,\ldots,x_{d-3})$. 
	We add $vw$ into $T$. 
	Note that the paths $P_{x, v}, P_{x, w}$ and the edge $vw$ create a unique cycle containing the edge $x_2x_3$.
	We remove the edge $x_2 x_3$ and obtain a spanning tree $\widehat{T}$.
	Note that $x_2$ and $x_3$ are two internal vertices of $P$ that are two additional leaves of $\widehat{T}$ that were not leaves in $T$.
	Every other internal vertex of $T$, apart from $x_2$ and $x_3$ has remained an internal vertex in $\widehat{T}$.
	In case $w$ is a leaf,  $w$ becomes an internal vertex in $\widehat{T}$, but $x_2$ and $x_3$ become two new leaves in $\widehat{T}$.
	Therefore, $\widehat{T}$ has at least one more leaf than $T$.
	In this case, $\widehat{T}$ has at least $s + 1$ leaves.
	Similarly, when $w$ is not a leaf, it follows that $x_2$ and $x_3$ are the only two additional leaves in $\widehat{T}$ that were not leaves in $T$.
	Also, $w$ remains an internal vertex in $\widehat{T}$ and every other internal vertex of $T$ remains internal in $\widehat{T}$.
	This implies that $\widehat{T}$ has at $s + 2$ leaves and $w$ is an internal vertex of $\widehat{T}$.
	Moreover, $x_2$ and $x_3$ are two internal vertices of $P$ that are two leaves in $\widehat{T}$.
	We refer to \Cref{fig:MLST-boundedness-case-1} for an illustration.
	\item[Case (ii):] $w$ is a descendant of $v$ in $T$. Suppose that $w$ is a descendant of $x_d$. Because $v$ is a strictly internal vertex of $P$, $(v,\ldots,x_{d-2},x_{d-1},x_d)$ is 	a subpath of the path $P_{v,w}$ between $v$ and $w$. Then we construct $\widehat{T}$ by adding  $vw$ and deleting $x_{d-1}x_{d-2}$. 	This way, we obtain a tree where both 	$x_{d-2}$ and $x_{d-1}$ become leaves. Moreover, only these vertices are new leaves of $\widehat{T}$. 
	The addition of $vw$ can decrease the number of leaves by at most one. Therefore, $\widehat{T}$ has at least $s+1$ leaves.
	Suppose than $w=x_i$ for some $i\in[d-1]$. Note that $x_i$ is an internal vertex of $T$. Then we construct $\widehat{T}$ by adding  $vw$ and deleting $x_{i-1}x_i$.  Since 		$P_{v,w}$ is of length at least two, this creates a tree with one additional leaf $x_{i-1}$. 
	We refer to the left part of \Cref{fig:lemma-boundedness-case-2} for an illustration.
	\item[Case (iii):] $w$ is an ancestor of $v$. Then $w=x_i$ for some $i\in[\ell]$. If $i\geq 2$ then we construct $\widehat{T}$ from $T$ by 		adding $vw$ and deleting $x_ix_{i+1}$.   Since $w$ is an internal vertex of $T$ and the path $P_{v,w}$ between $v$ and $w$ is of length at least two, $x_{i+1}$ becomes an additional leaf. If $w=x_1$, we construct $\widehat{T}$ by adding $vw$ and deleting $x_2x_3$. Similarly to Case (ii), both $x_2$ and $x_3$ are leaves of $\widehat{T}$ and only $w$ may become an internal vertex.     
	We refer to the right part  of \Cref{fig:lemma-boundedness-case-2} for an illustration.
\end{description}
The second part of the claim that  if $z$ is a leaf of $\widehat{T}$ that is not a leaf of $T$, then $z$ is an internal vertex of $P$ immediately follows from the above construction of $\widehat{T}$.
This completes the proof.
\end{proof}

\begin{figure}[t]
\centering
	\includegraphics[scale=0.3]{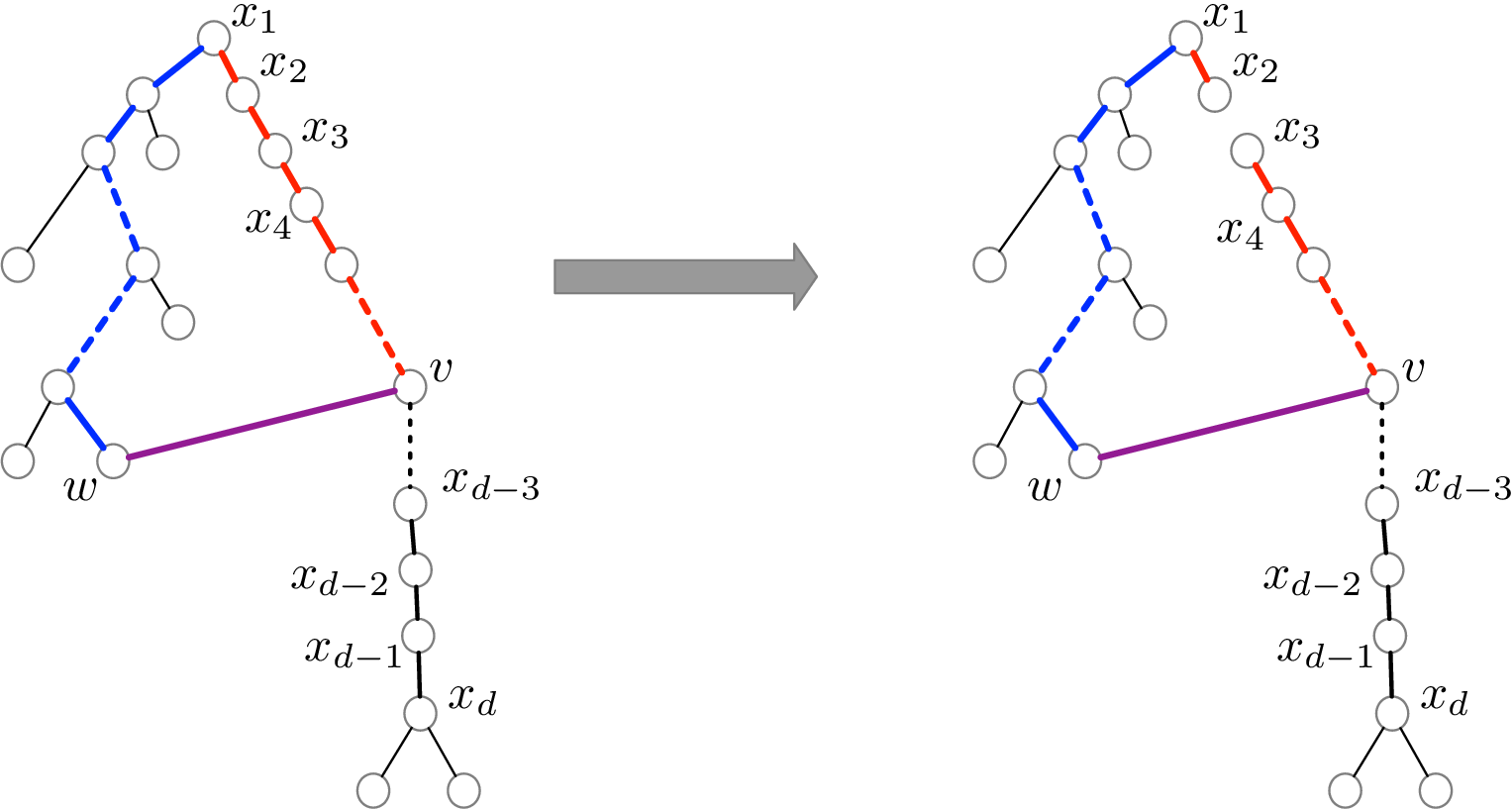}
	\caption{An illustration of \Cref{lemma:non-tree-edge-increase-leaves} proof Case (i). The tree on the left is $T$ that does not contain the pink purple colored edge $vw$ which is an edge of the graph. One of the red colored edges, $x_2 x_3$ from $T$ is removed  and the purple colored edge $vw$ has been added to create $\widehat{T}$.}
\label{fig:MLST-boundedness-case-1}	
\end{figure}

\begin{figure}[t]
\centering
	\includegraphics[scale=0.3]{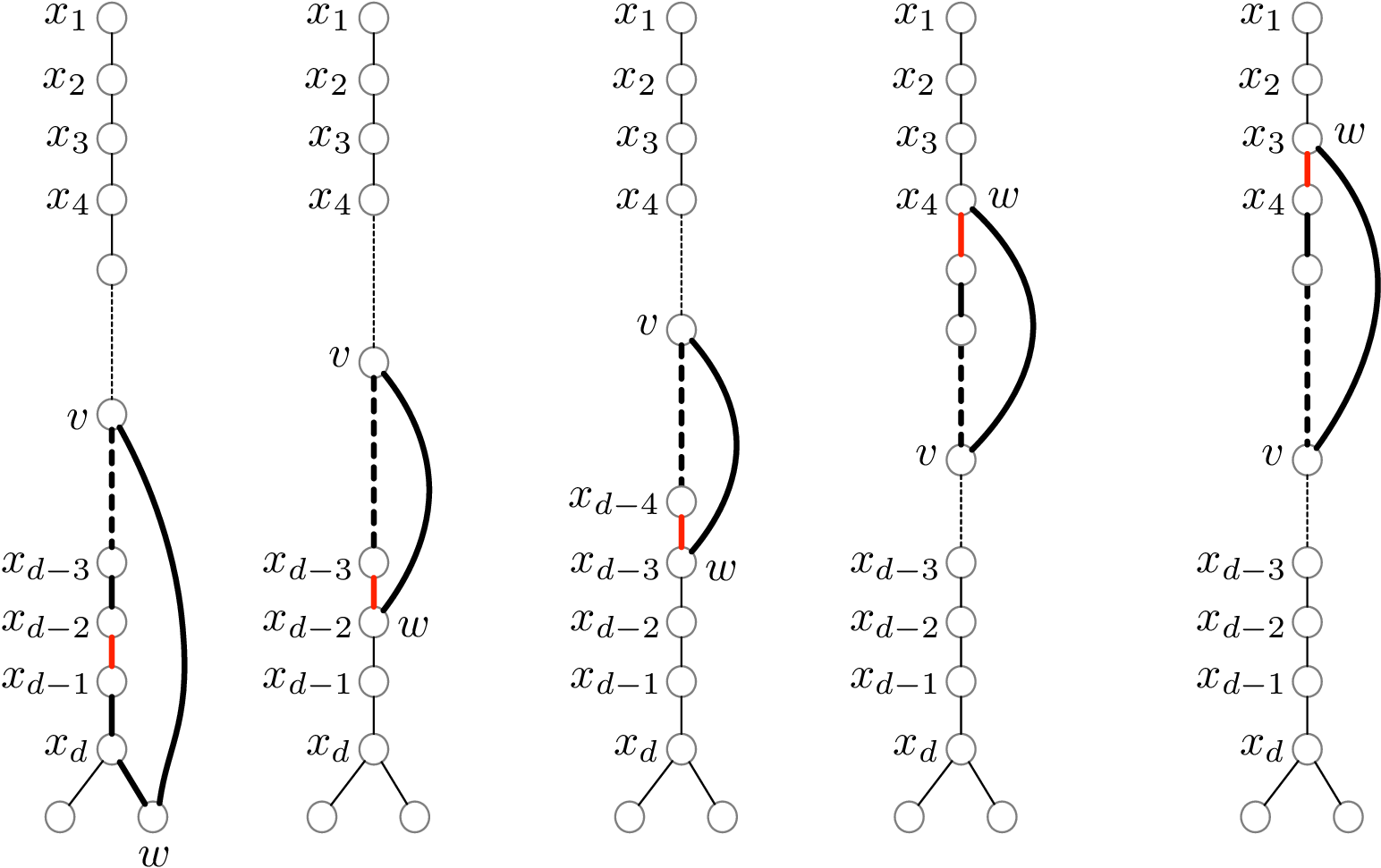}
	\caption{An illustration of Cases (ii) and (iii) of \Cref{lemma:non-tree-edge-increase-leaves} proof. It highlights what are the edges present in $T$ and the thick black colored edge $vw$ is not present in the tree. The red colored edge is removed from $T$ and $vw$ is added to create $\widehat{T}$.}
\label{fig:lemma-boundedness-case-2}
\end{figure}
\fi

We use this lemma to obtain the following lower bound for the maximum number of leaves in a spanning tree. 

\begin{lemma}
\label{lem:leavesbound}
Let $s\geq 2$ and $\ell\geq 1$ be integers. Let also  $G$ be a graph with a (possibly empty) subset of vertices $V_{\rm NT}\subseteq V(G)$ such that (i) $G$ has a spanning tree $T$ where each vertex of $V_{\rm NT}$ is internal and (ii) there is no  degree-$2$-path in $G$ of length at least $s$ with all internal vertices in $V(G)\setminus V_{NT}$.
Then 
	\begin{itemize}
	\item either $|V(G)|<(2\ell+|V_{\rm NT}|)(s+3)$
	\item or $G$ has a spanning tree $T$ with at least $\ell$ leaves such that  each vertex of $V_{\rm NT}$ is internal in $T$.
	\end{itemize}
\end{lemma} 

\begin{proof}
Suppose that $T$ is a spanning tree of $G$ with the maximum number of leaves such that each vertex of $V_{\rm NT}$ is internal, and assume that $T$ has at most $\ell-1$ leaves.
We prove that $|V(G)|<(2\ell+|V_{\rm NT}|)(s+3)$.  
Let $V_1$ and $V_{\geq 3}$ be the sets of vertices of degree one (i.e., leaves) and at least three in $T$, respectively. 
Consider $W=V_1\cup V_{\geq 3}\cup V_{\rm NT}$.
By Observation \ref{obs:leavesvsint}, $|W|\leq |V_1|+|V_{\geq 3}|+|V_{\rm NT}|\leq 2\ell-3+|V_{\rm NT}|$.  Thus, at least $|V(G)|-|V_{\rm NT}|-2\ell+3$ vertices of $G$ have degree at most two and do not belong to $V_{\rm NT}$. 
Notice that each of these vertices is an internal vertex of some degree-2-path in $T$.
Let $\mathcal{P}$ be the set of all inclusion maximal degree-2-paths in $T$ whose internal vertices are in $V(G)\setminus W$.
Then for each $v\in V(G)\setminus W$, there is $P\in \mathcal{P}$ such that $v$ is an internal vertex of $P$. 
Because $T$ is a tree and each $P\in\mathcal{P}$ has its end-vertices in $W$, we have  that $|\mathcal{P}| = |W| - 1 \leq 2\ell-4+|V_{\rm NT}|$. 
Let $P=(v_0,\ldots,v_r)\in \mathcal{P}$. Assume first that $r\geq 6$. 
{Consider the strictly internal vertices of $P$ that are $v_3,\ldots,v_{r-3}$.
As $r \geq 6$, it follows that $3 \leq r-3$.
Hence, we consider $P'=(v_2,\ldots,v_{r-2})$.
Observe that every internal vertex of $P'$ is strictly internal to $P$.
If an internal vertex $v$ of $P'$ has a neighbor $w$ such that $vw$ is not an edge of $T$, then due to Lemma \ref{lemma:non-tree-edge-increase-leaves}, there exists a spanning tree $T^*$ in $G$ such that $T^*$ has at least $\ell$ leaves. 
And more importantly, due to Lemma \ref{lemma:non-tree-edge-increase-leaves}, if there is a leaf $z \in T^*$ that is not a leaf in $T$, then $z$ is an internal vertex of $P$.
Note that $V_{\rm NT}$ is disjoint from the internal vertices of $P$, because for every $P \in \PP$, the internal vertices of $P$ are from $V(G) \setminus W$ and $W \supseteq V_{\rm NT}$.
This creates a spanning tree $T^*$ of $G$ that has more leaves than $T$, and the vertices of $V_{\rm NT}$ remain internal in $T^*$, contradicting the choice that $T$ is a spanning tree of $G$ with maximum possible number of leaves.}
Thus, $P'$ is a degree-2-path in $G$ with internal vertices outside $V_{\rm NT}$. 
Because there is no degree-$2$-path in $G$ of length at least $s$ with all internal vertices in $V(G)\setminus V_{NT}$, we have that $r< s+4 $ and $P$ has at most $s+2$ internal vertices. 
If $r<6$ then because $s\geq 2$, $P$ also has at most $s+2$ internal vertices. Thus, we have that every $P\in\mathcal{P}$ has at most $s+2$ internal vertices.
This implies that
\begin{equation*}
\begin{aligned}
|V(G)|\leq &|W|+|\mathcal{P}|(s+2)\leq 2\ell-3+|V_{\rm NT}|+(2\ell-4+|V_{\rm NT}|)(s+2)\\
=&(2\ell-4+|V_{\rm NT}|)(s+3)+1< (2\ell+|V_{\rm NT}|)(s+3).
\end{aligned}
\end{equation*}
This concludes the proof.
\end{proof}

We use the next lemma to argue that the existence of sufficiently many leaves in a spanning tree that have degrees at least two in the input graph guarantees the existence of a family of $k$-diverse spanning trees.

\begin{lemma}
\label{lem:diverse}
Let $k,\ell\geq 1$ be integers, and let $G$ be a graph with a given subset of vertices $V_{\rm NT}$. 
Let also $L$ be a set of vertices of degree at least two such that (i) $|L|\geq 2\lceil\frac{k}{4}\rceil\ell$, (ii) $G$ has a spanning tree $T$ such that every vertex of $V_{\rm NT}$ is internal,  
$L$ is a subset of the set of leaves of $T$, and each vertex of $V_{\rm NT}$ is adjacent to at least two vertices of  $V(G)\setminus L$ in $T$.
Then $G$ has  $\ell$ spanning trees $T_1,\ldots,T_\ell$ that are {\kDiv}.
Moreover, the vertices of $V_{\rm NT}$ are internal in each $T_i$.
Subsequently, if $T$ has $s$ leaves then for each $i\in[\ell]$, $T_i$ has at least $s-\lceil\frac{k}{4}\rceil$ leaves.  
\end{lemma}

\begin{proof}
For every vertex $v\in L$, denote by $p_v$ its unique neighbor in $T$.
Because every vertex $v\in L$ has degree at least two in $G$, $v$ has a neighbor $q_v$ in $G$ such that $q_v \neq p_v$. We assume that the vertices $q_v$ are selected in such a way that the total number of edges $uv\in E(G)$ such that $u=q_v$ or $v=q_u$ is minimum.
We show the following claim.

\begin{claim}
\label{cl:ind}
There is a set $L'\subseteq L$ of size at least $\lceil\frac{k}{4}\rceil\ell$ such that for every $v \in L'$,  
$q_v \notin L'$.
\end{claim}

\begin{claimproof}
Consider the graph $H$  such that
\begin{itemize}
	\item $V(H)=L$, and
	\item two distinct vertices $u,v\in V(H)$ are adjacent in $H$ if and only if $u=q_v$ or $v=q_u$.
\end{itemize}  
We claim that $H$ is a forest.

For the sake of contradiction, assume that $H$ has a cycle $C=(v_1,\ldots,v_r)$. Note that since $v_1,\ldots,v_r\in L$, $p_{v_i}\notin\{v_1,\ldots,v_r\}$ for $i\in[r]$. 
By symmetry, we assume without loss of generality that $q_{v_2}=v_1$.
Then because $v_2v_3\in E(H)$, $q_{v_3}=v_2$. 
Iterating, we obtain that $q_{v_1}=v_r$.
We modify the choice of $q_{v_2}$ and assign $q_{v_2}:=v_3$.
In this modified choice, observe that neither $q_{v_1}=v_2$ nor $q_{v_2}=v_1$, implying that $v_1v_2 \notin E(H)$.
Thus, we decreased the total number of edges $uv\in E(G)$ such that   $u=q_v$ or $v=q_u$.
However, this contradicts the choice of $q_v$ for $v\in L$.
Therefore, $H$ is a forest.
Since $H$ is a forest, $H$ is bipartite.
Therefore, $H$ has an independent set $L'$ of size at least $\frac{|V(H)|}{2}=\frac{|L|}{2}\geq \lceil\frac{k}{4}\rceil\ell$. 
By the definition of  $H$, we obtain that $q_v\notin L'$ for each $v\in L'$. This concludes the proof.
\end{claimproof}

Consider the set $L'$ which exists by \Cref{cl:ind}. Because $|L'|\geq \lceil\frac{k}{4}\rceil\ell$, there are $\ell$ disjoint subsets $L_1,\ldots,L_\ell\subseteq L'$ such that $|L_i|=\lceil\frac{k}{4}\rceil$ for each $i\in[\ell]$.  For every $i\in[\ell]$, we define $T_i$ be a spanning subgraph of $G$ with $E(T_i)=(E(T)\cup\{vq_v\mid v\in L_i\})\setminus \{vp_v\mid v\in L_i\}$.
 Notice that constructing $T_i$ can be seen as modifying $T$ by consecutively adding the edge $vq_v$ and deleting $vp_v$ for $v\in L_i$.
 Then applying inductive arguments, it can be seen that the graph obtained in each step is a tree, as the addition of $vq_v$ creates the unique cycle containing $vq_v$, which is destroyed by deleting $vp_v$. 
 Thus, we conclude that for  each $i\in[\ell]$, $T_i$ is a tree. 
 Because $L_i$ and $L_j$ are disjoint and $vq_v\neq uq_u$ for $v\in L_i$ and $u\in L_j$, for distinct $i,j\in[\ell]$, $|E(T_i)\triangle E(T_j)|=2(|L_i|+|L_j|)\geq 4\lceil\frac    {k}{4}\rceil\geq k$. Thus, $T_1,\ldots,T_\ell$ are {\kDiv} spanning trees of $G$. 
 Since each vertex of $V_{\rm NT}$ has at least two neighbors outside $L$, we have that the vertices of $V_{\rm NT}$ are internal in every $T_i$ ($1 \leq i \leq \ell$).

 To see the second claim of the lemma, 
 by assumption, $T$ has $s$ leaves.
 Notice that in each $T_i$ for $i\in[\ell]$, the vertices of $L_i$ are leaves of $T_i$. However, the vertices $q_v$ for $v\in L_i$ may be leaves of $T$ which become internal vertices of $T_i$.
 However, we have at most $\lceil\frac{k}{4}\rceil$ such vertices.
 Then  for each $i\in[\ell]$, $T_i$ has at least $s-\lceil\frac{k}{4}\rceil$ leaves where $s$ is the number of leaves of $T$. This concludes the proof.   
\end{proof}

The next lemma is used for reducing long degree-$2$-paths.

\begin{lemma}
\label{lem:long}
Let $G$ be a graph with a given subset of vertices $V_{\rm NT}\subseteq V(G)$. Let also $p,q\geq 0$ and $k,\ell\geq 1$ be integers.
Suppose that $G$ contains a degree-$2$-path $P$ of length at least $\ell+3$ whose internal vertices {are disjoint from} $V_{\rm NT}$ and denote 
by $G'$ the graph obtained by the contraction of an edge of $P$. Then the following two statements are equivalent:
\begin{enumerate}
\item $G$ has $\ell$ {\kDiv} spanning trees such that 
\begin{itemize}
	\item[(i)] the vertices of $V_{\rm NT}$ are internal vertices of each tree, and 
	\item[(ii)] every tree has at least $p$ leaves and at least $q$ internal vertices
\end{itemize}
\item  $G'$ has $\ell$ {\kDiv} spanning trees such that
\begin{itemize}
	\item[(i)] the vertices of $V_{\rm NT}$ are internal vertices of each tree,
	\item[(ii$^*$)] every tree has at least $p$ leaves and at least $\max\{0,q-1\}$ internal vertices.
\end{itemize}
\end{enumerate}
\end{lemma}

\begin{proof}
Suppose that $P=(v_0,\ldots,v_s)$ for some $s\geq \ell+3$. 

For the forward direction ($\Rightarrow$), assume that 
$T_1,\ldots,T_{\ell}$ are {\kDiv} spanning trees of $G$ satisfying (i) and (ii).
Notice that at most one edge from $\{v_1 v_2,\ldots,v_{s-2}v_{s-1}\}$  is not an edge of $T_i$. 
As $s \geq \ell + 3$ and there are $s-2$ edges in $\{v_1 v_2,\ldots,v_{s-2}v_{s-1}\}$, hence, there are at least $\ell + 1$ edges in $\{v_1 v_2,\ldots,v_{s-2}v_{s-1}\}$.
Therefore, there is an edge $e$ from $\{v_1 v_2,\ldots,v_{\ell + 1}v_{\ell + 2}\}$ that is present in each of the trees $T_1,\ldots,T_{\ell}$.

As $P$ is a degree-2-path in $G$, we can assume without loss of generality that $G'$ is obtained by contracting the edge $e$ 
(note that regardless of which edge from $P$ is contracted, the resulting graph falls in the same isomorphism class).
Let $T_1',\ldots,T_\ell'$ be the graph obtained from $T_1,\ldots,T_\ell$, respectively, by contracting an edge $e$ from $P$. 
Notice that $T_1',\ldots,T_\ell'$ are spanning trees of $G'$ because these trees and $G'$ were obtained by contracting the same edge $e$.
Also, because $e$ is an edge of every tree, we have that 
$|E(T_h')\triangle E(T_i')|=|E(T_h)\triangle E(T_i)|$ for all $h,i\in[\ell]$.
Then $T_1',\ldots,T_\ell'$ are {\kDiv} spanning trees.
Since the endpoints of $e$ are internal vertices of $P$, the vertices of $V_{\rm NT}$ are internal vertices of each tree $T_1',\ldots,T_\ell'$ and (i) holds.  To see (ii$^*$), consider an arbitrary $i\in[\ell]$. As $e=v_{j-1}v_j\in E(T_i)$ for $j\in\{2,\ldots,\ell+2\}$ and $s\geq \ell+3$, we have that either $v_{j-2}v_{j-1},v_{j}v_{j+1}\in E(T_i)$, or $v_{j-2}v_{j-1}\in E(T_i)$ and $v_{j}v_{j+1}\notin E(T_i)$, or $v_{j-2}v_{j-1}\notin E(T_i)$ and $v_{j}v_{j+1}\in E(T_i)$. In all these cases, the contraction of $e$ reduces the number of internal vertices by one but does not change the number of leaves. This implies (ii${^*}$).

For the opposite direction ($\Leftarrow$), assume without loss of generality that $G'$ was obtained from $G$ by contracting $e=v_1v_2$, and denote by $u$ the vertex obtained from $v_1$ and $v_2$ by the contraction of edge $v_1 v_2$.
Let 
$T_1',\ldots,T_\ell'$ be {\kDiv} spanning trees of $G'$ satisfying (i) and (ii$^*$). We construct the spanning trees $T_1,\ldots,T_\ell$ from $T_1',\ldots,T_\ell'$, respectively, by ``uncontracting'' $u$. Formally, let $i\in[\ell]$. Notice that $s\geq\ell+3\geq 3$.
Then three cases arise concerning the degree-2-path $(v_0, u, v_3)$ in $G'$.
In the first case, $T_i'$ contains the path $(v_0,u, v_3)$.
Then, we replace $(v_0,u,v_3)$ in $T_i'$ by $(v_0,v_1,v_2,v_3)$ to construct  $T_i$.
In the second case, $v_0u\in E(T_i')$ and $uv_3\notin E(T_i')$.
Then, we replace the edge $v_0u$ by the path $(v_0,v_1,v_2)$.
In the third case, $v_0u\notin E(T_i')$ and $uv_3\in E(T_i')$. 
Then, we replace the edge $uv_3$ by $(v_1,v_2,v_3)$.

For every spanning tree $T'$ of $G'$, there is at most one edge from $(v_0, u, v_3)$ that is not part of $T'$, therefore, these cases are exhaustive.

Because the vertices of $V_{\rm NT}$ are internal vertices of $T_i'$, and due to (i), $V_{\rm NT}$ is disjoint from the internal vertices of $P'$, we have that the vertices of $V_{\rm NT}$ are internal vertices of $T_i$ and disjoint from the internal vertices of $P$.
Thus, (i) holds. For (ii), it is sufficient to observe that the construction of $T_i$ from $T_i'$ does not decrease the number of leaves and increases the number of internal vertices by one. This completes the proof.
\end{proof}

\section{Kernelization Algorithms}\label{sec:kern}
In this section, we use the technical lemmata proved in the previous section to give a proof of our kernelization results.

\subsection{Kernel for {\DivLIST}}

First, we consider \dLIST and prove \Cref{thm:kernLIST}, which we restate.

{\LIST*}

\begin{proof}
Let $(G,p,q,k,\ell)$ be an instance of {\dLIST}.
For this problem, it is implicitly assumed that $V_{\rm NT} = \emptyset$ here.
We assume that $G$ is a connected graph.
Otherwise, $G$ has no spanning tree, and we can immediately return a trivial no-instance. 
We also assume that $G$ is not a tree.
Otherwise, if $G$ is a tree, then $G$ has a unique spanning tree.
It is straightforward to solve {\dLIST} and return either a trivial yes or no-instance. 

So, we are in the case that $G$ is a connected graph and $G$ is not a tree.
Next, we exhaustively apply the following two reduction rules.

\begin{reductionrule}
\label{rule:contract}
If $G$ contains a degree-$2$-path $P$ of length at least $\ell+3$ then contract an edge of $P$ and set $q=\max\{0,q-1\}$.
\end{reductionrule}

As $V_{\rm NT}=\emptyset$, safety of the rule follows from \Cref{lem:long}.
For the next rule, notice that a pendant edge is included in any spanning tree.
Furthermore, since $G$ is not a tree, the unique neighbor of a pendant vertex is an internal vertex of every spanning tree.
These observations imply the safeness of the following rule.

\begin{reductionrule}
\label{rule:twoleaves}
If $G$ contains two pendant vertices $x$ and $y$ adjacent to the same neighbor then set $G:=G-x$ and $p=\max\{0,p-1\}$.
\end{reductionrule}

Denote by $h$ the number of pendant vertices  of $G$. 
Notice that after the exhaustive application of \Cref{rule:twoleaves}, any two pendant vertices have distinct neighbors.
If $G$ is an irreducible instance with respect to \Cref{rule:contract} and \Cref{rule:twoleaves} and has $h$ pendant vertices, then the total number of the neighbors of pendant vertices is $h$.
Notice that all pendant edges are included in every spanning tree of $G$.
If $G$ has at least $p$ pendant vertices (i.e. $h \geq p$), then every spanning tree of $G$ has at least $p$ leaves.
Also, because the neighbors of pendant vertices are internal in any spanning tree, the other endpoint of an edge that is not a pendant vertex is an internal vertex in any spanning tree of $G$.
Hence, if $G$ has at least $q$ vertices that are neighbors of pendant vertices, then each spanning tree has at least $q$ internal vertices. 
The above two facts allow us to apply the next rule. 

\begin{reductionrule}
\label{rule:leaves}~
\begin{itemize}
\item If the number of pendant vertices $h\geq p$ then set $p:=0$, and 
\item if $h\geq q$ then set $q:=0$.
\end{itemize}
\end{reductionrule}

Note that the application of the above reduction rules can produce distinct combinations of values of $p$ and $q$.
We divide these combinations into two cases.
The first case is when both $p = q = 0$ and the other case is when one of $p$ and $q$ is positive, i.e. $\max(p, q) > 0$.
Depending on the two cases, we design the subsequent reduction rules as follows.

\subparagraph*{Case~1:} 
In this case, \Cref{rule:leaves} assigns $p=q=0$. 
Then, we have no constraint on the number of leaves and the number of internal vertices in a spanning tree.
We can use the observation that every pendant edge is included in any spanning tree and exhaustively apply the following rule.

\begin{reductionrule}
\label{rule:delleaves}
If $G$ has a pendant vertex $v$ then set $G:=G-v$.
\end{reductionrule}

If $v$ is a pendant vertex and $uv$ is a pendant edge of $G$, then $uv$ is part of any spanning tree of $G$.
Then, the safeness of this reduction rule follows because $uv$ is an edge for any collection $T_1,\ldots,T_{\ell}$ of {\kDiv} spanning trees.
Deleting $v$ from $T_i$ creates $T_i'$ and we observe that $|E(T_i') \triangle E(T_j')| = |E(T_i) \triangle E(T_j)|$.
Hence, $T_1',\ldots,T_{\ell}'$ is a collection of {\kDiv} spanning trees, ensuring the forward direction of the safeness of the above reduction rule.
The backward direction can also be proved using similar arguments.

After applying the above-mentioned reduction rules, we have that $G$ has no pendant vertex.
Also, because the input graph is not a tree, the contraction of edges of degree-2-paths of length at least $3$ and 
the deletion of pendant vertices result in a graph of minimum degree of at least two. Thus, the minimum degree of $G$ is at least two. 
Then we apply the following reduction rule.

\begin{reductionrule}\label{rule:simplefinal}
If $|V(G)|<4\lceil\frac{k}{4}\rceil\ell(\ell+6)$ then return the instance $(G,p,q,k,\ell)$ and stop.
If $|V(G)| \geq 4\lceil\frac{k}{4}\rceil\ell(\ell+6)$, then return a trivial yes-instance.
\end{reductionrule}

To argue that the rule is safe, assume that $|V(G)|\geq 4\lceil\frac{k}{4}\rceil\ell(\ell+6)$.
We use the fact that after the exhaustive
application of \Cref{rule:contract}, $G$ has no degree-$2$-path of length $\ell+3$.
Now, we apply \Cref{lem:leavesbound} with $V_{\rm NT}=\emptyset$.
As $G$ has at least $4\lceil\frac{k}{4}\rceil\ell(\ell+6)$ vertices, it follows that $G$ has a spanning tree having at least $2\lceil\frac{k}{4}\rceil \ell$ leaves.

Let $L$ be the set of these $2\lceil\frac{k}{4}\rceil \ell$ leaves of $T$.
Because every vertex of $L$ has degree at least two in $G$, by applying \Cref{lem:diverse} for $V_{\rm NT}=\emptyset$, we have that $G$ has  $\ell$ {\kDiv} spanning trees $T_1,\ldots,T_\ell$.
Thus, $(G,p,q,k,\ell)$ is a yes-instance of {\dLIST} if $|V(G)|\geq 4\lceil\frac{k}{4}\rceil\ell(\ell+6)$.
This concludes Case~1.

\subparagraph*{Case~2.} In this case, $\max\{p,q\}>0$.
We have that \Cref{rule:leaves} assigns $p \neq 0$ or $q \neq 0$. 
Then, it must be that $h < p$ or $h < q$, implying that $h < \max(p, q)$.
In this case, we use the following variant of \Cref{rule:simplefinal}.

\begin{reductionrule}\label{rule:advfinal}
If $|V(G)|<(2(\max\{p,q\}+2\lceil\frac{k}{4}\rceil\ell)+q)(\ell+6)$ then return the instance $(G,p,q,k,\ell)$ and stop. Otherwise, call the algorithm from \Cref{prop:MIST-dec-result} for the instance $(G,q)$ of \MIST, and for its output $(G',q')$, output the instance $(G',0,q',1,1)$ of \dLIST.
\end{reductionrule}

To show safeness, assume that $|V(G)|\geq (2(\max\{p,q\}+2\lceil\frac{k}{4}\rceil\ell)+q)(\ell+6)$.
First, we argue that $(G,p,q,k,\ell)$ is a yes-instance of {\dLIST} if and only if $(G,0,q,1,1)$ is a yes-instance of {\dLIST}.

The forward implication ($\Rightarrow$) is trivial.
Thus, we have to prove the backward direction ($\Leftarrow$) of the reduction rule.
We assume that $(G,0,q,1,1)$ is a yes-instance to {\dLIST}.
We argue that the same holds for $(G,p,q,k,\ell)$.

If $(G,0,q,1,1)$ is a yes-instance to {\dLIST}, then $G$ has a spanning tree $T$ with at least $q$ internal vertices. 
Because $|V(G)|\geq (2(\max\{p,q\}+2\lceil\frac{k}{4}\rceil\ell)+q)(\ell+6)$ we can apply  \Cref{lem:leavesbound} by choosing an arbitrary set $V_{\rm NT}$ of $q$ internal vertices of $T$.
As \Cref{rule:contract} is not applicable, $G$ has no degree-2-path of length at least $\ell + 3$.
Hence, $G$ has no degree-2-path of length at least $\ell + 3$ whose internal vertices are disjoint from $V_{\rm NT}$.
Therefore, we can use the premise of \Cref{lem:leavesbound}. 
Then $G$ has a spanning tree $T'$ with at least $\max\{p,q\}+2\lceil\frac{k}{4}\rceil\ell$ leaves.

We redefine $V_{\rm NT}$ and $L$ as follows.
Recall that because of \Cref{rule:leaves}, $G$ has $h<\max\{p,q\}$ pendant vertices.
 If $T'$ has at least $2\lceil\frac{k}{4}\rceil\ell+h+2q$ leaves, then we first choose $V_{\rm NT}$ be an arbitrary set of $q$ internal vertices of $T'$.
 For each $v\in V_{\rm NT}$, we choose  two arbitrary neighbors $x_v$ and $y_v$ in $T'$, and set $X=\bigcup_{v\in V_{\rm NT}}\{x_v,y_v\}$. 
 Since $|X|\leq 2q$, and $T'$ has at least $(2\lceil\frac{k}{4}\rceil\ell+h+2q)$ leaves, no matter how we choose $X$, there are at at least $2\lceil\frac{k}{4}\rceil\ell+h$ leaves in $T'$ that are not in $X$.
 Out of these $2\lceil\frac{k}{4}\rceil\ell+h$ leaves in $T'$, $h$ are pendant vertices in $G$.
 We choose the rest of $2\lceil\frac{k}{4}\rceil\ell$ such leaves of $T'$ as $L$, that are not pendant vertices of $G$.
 In particular, by choice of $L$, every vertex of $V_{\rm NT}$ has at least two neighbors in $T$ that are in $X$ such that $X \cap L = \emptyset$.
 Hence, we get a set $L$ such that $|L| \geq 2\lceil\frac{k}{4}\rceil\ell$, and a set $V_{\rm NT}$ of internal vertices $T'$ such that every vertex of $V_{\rm NT}$ has at least two neighbors outside $L$.
 This ensures us that the assumption to \Cref{lem:diverse} is satisfied.
 
Suppose that $T'$ has less than  $2\lceil\frac{k}{4}\rceil\ell+h+2q$ leaves.
Note that in this case, less than $2\lceil\frac{k}{4}\rceil\ell+h+2q$ vertices of $T'$ are the neighbors of leaves. 
Because $|V(G)|\geq (2(\max\{p,q\}+2\lceil\frac{k}{4}\rceil\ell)+q)(\ell+6)$, at least 
$(2(\max\{p,q\}+2\lceil\frac{k}{4}\rceil\ell)+q)(\ell+6)- 2(2\lceil\frac{k}{4}\rceil\ell+h+2q)$ vertices of $T'$ are neither leaves nor neighbors of the leaves.
Note that $(2(\max\{p,q\}+2\lceil\frac{k}{4}\rceil\ell)+q)(\ell+6)- 2(2\lceil\frac{k}{4}\rceil\ell+h+2q) \geq 12q + 8k\ell + k\ell^2 \geq q$.
Hence, there are at least $q$ vertices in $T'$ that are neither leaves nor the neighbors of leaves.
Therefore, there are $q$ vertices that are internal to $T'$.
As they are internal in $T'$, we choose $V_{\rm NT}$ to be a set of $q$ such vertices non-adjacent to the leaves.
As there are $2\lceil\frac{k}{4}\rceil\ell + \max(p, q) > 2\lceil\frac{k}{4}\rceil\ell + h$ leaves in $T'$, we can choose any set $L$ of $2\lceil\frac{k}{4}\rceil\ell$ leaves of $T'$ that are not pendants of $G$.
This ensures us that the premise to \Cref{lem:diverse} is satisfied.

Then by  \Cref{lem:diverse}, $G$ has $\ell$ {\kDiv} spanning trees $T_1,\ldots,T_\ell$ such that for each $i\in[\ell]$, the vertices of $V_{\rm NT}$ are internal in $T_i$, and $T_i$ has at least 
 $\max\{p,q\}+2\lceil\frac{k}{4}\rceil\ell-\lceil\frac{k}{4}\rceil\geq p$ leaves.
 This means that  $(G,p,q,k,\ell)$ is a yes-instance  of {\dLIST}.
 
We have that the instances  $(G,p,q,k,\ell)$ and $(G,0,q,1,1)$ of {\dLIST} are equivalent. Observe that the instance $(G,0,q,1,1)$ of {\dLIST} is equivalent to the instance $(G,q)$ of {\MIST}.  
Therefore, we can apply the kernelization algorithm from~\Cref{prop:MIST-dec-result}. Because for its output $(G',q')$, the instance $(G',q')$ of {\MIST} is equivalent to the instance 
$(G',0,q',1,1)$ of \dLIST, we conclude that \Cref{rule:advfinal} is safe. This concludes Case~2 and the description of the kernelization algorithm for \dLIST.

The correctness of the algorithm follows from the safeness of \Cref{rule:contract}--\Cref{rule:advfinal}.
To upper-bound the size of the output graph, observe that in the worst case, we output a graph with $|V(G)|<(2(\max\{p,q\}+2\lceil\frac{k}{4}\rceil\ell)+q)(\ell+6)$ in \Cref{rule:advfinal} as the kernelization algorithm from~\Cref{prop:MIST-dec-result} produces a graph with at most $2q$ vertices. Thus, we obtain a kernel with $\OO((p+q+k\ell)\ell)$ vertices. Since \Cref{rule:contract}--\Cref{rule:advfinal} can be applied in polynomial time, the overall running time of our algorithm is polynomial. This completes the proof.
\end{proof}

Observe that in the above problem {\dLIST}, if we put $p = 0$, then the problem is the same as {\DivMIST}, that is, the diverse variant of \MIST. 
Analogously, if we put $q = 0$, then the problem is the same as {\DivMLST}---the diverse variant of \MLST.
Hence, we have the following results as a corollary.

\begin{corollary}
\label{cor:MLST-MIST}
{\DivMLST} parameterized by $p + k + \ell$ admits a kernel with $\OO((p+k\ell)\ell)$ vertices.
{\DivMIST} parameterized by $q + k + \ell$ admits a kernel with $\OO((q+k\ell)\ell)$ vertices.
\end{corollary}

\subsection{Kernel for {\DivLNTST}}

We restate \Cref{thm:kernLNTST} here, which is proved by similar arguments.
Thus, we only sketch its proof, focusing on the differences between the algorithms for {\dLIST} and {\dLNTST}.

\LNPST*

\begin{proof}
Let $(G,V_{\rm NT},p,k,\ell)$ be an instance of {\dLNTST}. We assume that $G$ is a connected graph distinct from a tree, as otherwise, {\dLNTST} has a straightforward solution. Also, if there is $v\in V_{\rm NT}$ with $\deg_G(v)=1$, then we immediately return a trivial no-instance and stop as a pendant vertex of $G$ is a leaf of any spanning tree.  
Next, we exhaustively apply the following reduction rule whose safeness follows from \Cref{lem:long}. 

\begin{reductionrule}
\label{rule:contractNT}
If $G$ contains a degree-$2$-path $P$ of length at least $\ell+3$ the internal vertices of which are disjoint from $V_{\rm NT}$, then contract an edge of $P$.
\end{reductionrule}

Since all pendant edges of $G$ are included in every spanning tree, if $G$ has at least $p$ pendant vertices, then every spanning tree has at least $p$ leaves.
Therefore, we apply the next rule. 

\begin{reductionrule}
\label{rule:leavesNT}
If $G$ has at least $p$ pendant vertices, then set $p:=0$.
\end{reductionrule}

We proceed with considering two cases depending on the value of $p$.

\subparagraph*{Case~1.} In this case, we consider $p=0$.
We use the observation that every pendant edge is included in any spanning tree, and each vertex adjacent to a pendant vertex is an internal vertex in any spanning tree.
This allows us to get rid of pendant vertices by exhaustively applying the following rule.

\begin{reductionrule}
\label{rule:delleavesNT}
If $G$ has a pendant vertex $v$ then set $G:=G-v$ and set $V_{\rm NT}:=V_{\rm NT}\setminus\{u\}$ for the unique neighbor $u$ of $v$ in $G$.
\end{reductionrule}

After applying the rule, we have that $G$ has no pendant vertices. Also, because $G$ is not a tree, the minimum degree of $G$ is at least two. 
Then we apply the following rule.  

\begin{reductionrule}\label{rule:simplefinalNT}
If $|V(G)|< (4\lceil\frac{k}{4}\rceil\ell+5|V_{\rm NT}|)(\ell+6)$ then return the instance $(G,V_{\rm NT},p,k,\ell)$ and stop. 
Otherwise, call the algorithm from \Cref{prop:NTST-dec-result} for the instance $(G,V_{\rm NT})$ of \NTST, and for its output $(G',V_{\rm NT}')$, output the instance $(G',V_{\rm NT}',0,1,1)$ of \dLNTST.
\end{reductionrule}

To prove that the rule is safe, assume that $|V(G)|\geq (4\lceil\frac{k}{4}\rceil\ell+5|V_{\rm NT}|)(\ell+6)$.
We claim that in this case, the instances 
$(G,V_{\rm NT},p,k,\ell)$ and $(G,V_{\rm NT},0,1,1)$ of {\dLNTST} are equivalent.

For the forward implication ($\Rightarrow$) is straightforward that if $(G,V_{\rm NT},0,k,\ell)$ is a yes-instance, then $(G,V_{\rm NT},0,1,1)$ is a yes-instance as well.

For the backward direction ($\Leftarrow$), we assume that $(G,V_{\rm NT},0,1,1)$ is a yes-instance.
Then $G$ has a spanning tree $T$ such that the vertices of $V_{\rm NT}$ are internal in $T$.
As \Cref{rule:contractNT} is not applicable, $G$ has no degree-$2$-path of length at least $\ell+3$ the internal vertices of which are disjoint from $V_{\rm NT}$.
Observe that $(4\lceil\frac{k}{4}\rceil\ell+5|V_{\rm NT}|)(\ell+6) = (2(2\lceil\frac{k}{4}\rceil\ell+2|V_{\rm NT}|) + |V_{\rm NT}|)(\ell + 6)$.
Then, by putting $s = \ell + 3$ in Lemma \ref{lem:leavesbound}, we conclude that $G$ has a spanning tree $T'$ with at least $2\lceil\frac{k}{4}\rceil\ell+2|V_{\rm NT}|$ leaves, and every vertex of $V_{\rm NT}$ is internal in $T'$. 

For every $v\in V_{\rm NT}$, we arbitrarily select two distinct neighbors  $x_v$ and $y_v$ in $T'$.
As $v \in V_{\rm NT}$ is an internal vertex in $T'$, these neighbors $x_v$ and $y_v$ exist, and there are at most $2|V_{\rm NT}|$ such neighbors.
Since the vertices of $V_{\rm NT}$ are internal in $T'$, $G$ has no pendant vertex because \Cref{rule:delleavesNT} is not applicable, and $T'$ has at least  $2\lceil\frac{k}{4}\rceil\ell + 2|V_{\rm NT}|$ leaves, there are at least $2\lceil\frac{k}{4}\rceil\ell$ leaves of $T'$ that are disjoint from these neighbors $x_v$ and $y_v$.
In particular, each of these $2\lceil\frac{k}{4}\rceil\ell$ leaves of $T'$ is of degree at least two in $G$.

We denote by $L$ the set of these leaves that are disjoint from these $2|V_{\rm NT}|$ many vertices that are neighbors of the vertices in $V_{\rm NT}$.
Observe that $|L| \geq 2\lceil\frac{k}{4}\rceil\ell$, $L$ is a subset of leaves in $T'$, and every vertex of $V_{\rm NT}$ is adjacent to at least two vertices of $T'$ that are not in $V_{\rm NT}$.
This spanning tree $T'$ together with $L$ and $V_{\rm NT}$ satisfy the preconditions to 
\Cref{lem:diverse}.
As $T'$ has $2\lceil\frac{k}{4}\rceil\ell+2|V_{\rm NT}|$ leaves, 
Now, we use \Cref{lem:diverse} and deduce that $G$ has $\ell$ {\kDiv} spanning trees $T_1,\ldots,T_\ell$ such that the vertices of $V_{\rm NT}$ is internal in each $T_i$.
 Therefore,  $(G,V_{\rm NT},0,k,\ell)$ is a yes-instance of \dLNTST. This proves that the rule is safe.

\medskip
We use the observation that the instance  $(G,V_{\rm NT},0,1,1)$ of {\dLNTST} is equivalent to the instance $(G,V_{\rm NT})$ of {\NTST}.
With this observation, we combine the fact that the instances $(G,V_{\rm NT},0,k,\ell)$ and $(G,V_{\rm NT},0,1,1)$ of {\dLNTST} are equivalent. 
Next, we invoke the kernelization algorithm from Proposition \ref{prop:NTST-dec-result}.
Because its output $(G',V_{\rm NT}')$ is equivalent to the instance $(G',V_{\rm NT}',0,1,1)$ of \dLNTST, we obtain that \Cref{rule:simplefinalNT} is safe. 
This completes our kernelization algorithm for Case~1 ($p = 0$).

\subparagraph*{Case~2.} In this case, we consider that $p>0$.
Notice that \Cref{rule:leavesNT} is not applicable, and $p > 0$, it implies that the number of pendant vertices of $G$ is less than $p$.
In this case, our next reduction rule is the following.

\begin{reductionrule}\label{rule:advfinalNT}
If $|V(G)|< (4\lceil\frac{k}{4}\rceil\ell+2p+5|V_{\rm NT}|)(\ell+6)$ then return the instance $(G,V_{\rm NT},p,k,\ell)$ and stop. 
Otherwise, call the algorithm from \Cref{prop:NTST-dec-result} for the instance $(G,V_{\rm NT})$ of \NTST, and for its output $(G',V_{\rm NT}')$, output the instance $(G',V_{\rm NT}',0,1,1)$ of \dLNTST.
\end{reductionrule}

The safeness is proved similarly to the safeness of \Cref{rule:simplefinalNT}. Suppose that  $|V(G)|\geq (4\lceil\frac{k}{4}\rceil\ell+2p+5|V_{\rm NT}|)(\ell+6)$. We show that  the instances 
$(G,V_{\rm NT},p,k,\ell)$ and $(G,V_{\rm NT},0,1,1)$ of {\dLNTST} are equivalent. 

For the forward implication ($\Rightarrow$), if $(G,V_{\rm NT},p,k,\ell)$ is a yes-instance, then $(G,V_{\rm NT},0,1,1)$ is a yes-instance.
For the backward direction ($\Leftarrow$), we assume that $(G,V_{\rm NT},0,1,1)$ is a yes-instance and $G$ has a spanning tree $T$ such that the vertices of $V_{\rm NT}$ are internal in $T$.
Because \Cref{rule:contractNT} cannot be applied to the considered instance, $G$ has no degree-$2$-path of length at least $\ell+3$ having internal vertices {disjoint} from $V_{\rm NT}$.
As {$|V(G)|\geq(4\lceil\frac{k}{4}\rceil\ell+2p+5|V_{\rm NT}|)(\ell+6) = (2(2\lceil\frac{k}{4}\rceil\ell + p + 2|V_{\rm NT}|) + |V_{\rm NT}|)(\ell + 6)$,}
due to \Cref{lem:leavesbound}, $G$ has a spanning tree $T'$ with at least $2\lceil\frac{k}{4}\rceil\ell+p+2|V_{\rm NT}|$ leaves such that every vertex of $V_{\rm NT}$ is internal in $T'$.
For every $v\in V_{\rm NT}$, we arbitrarily select two distinct neighbors  $x_v$ and $y_v$ in $T'$.
There are at most $2|V_{\rm NT}|$ such neighbors $x_v$ and $y_v$.
Since the vertices of $V_{\rm NT}$ are internal in $T'$, and $T'$ has at least  $2\lceil\frac{k}{4}\rceil\ell+p + 2|V_{\rm NT}|$  leaves, there are at least $2\lceil\frac{k}{4}\rceil\ell + p$ leaves in $T'$ that are disjoint from these neighbors $x_v$ and $y_v$.
{But $G$ has at most $p-1$ pendant vertices, and $T'$ has at least $2\lceil\frac{k}{4}\rceil\ell + p$ leaves in $T'$ that are disjoint from these neighbors $x_v$ and $y_v$.
Out of these $2\lceil\frac{k}{4}\rceil\ell + p$ leaves, there are at most $p-1$ pendant vertices.
Hence, $T'$ has $2\lceil\frac{k}{4}\rceil\ell$ leaves that have degree at least two in $G$, and every vertex $v$ from $V_{\rm NT}$ has at least two neighbors $x_v, y_v$ in $T'$ that are disjoint from these $2\lceil\frac{k}{4}\rceil\ell$ leaves.}
We denote these $2\lceil\frac{k}{4}\rceil\ell$ leaves by $L$.
Observe that every vertex of $V_{\rm NT}$ has at two neighbors that are outside $L$.

Applying \Cref{lem:diverse}, we obtain that $G$ has $\ell$ {\kDiv} spanning trees $T_1,\ldots,T_\ell$ such that the vertices of $V_{\rm NP}$ are internal in each tree.
Furthermore, as $T'$ has $2\lceil\frac{k}{4}\rceil\ell+p+2|V_{\rm NT}|$ leaves, for each $i\in[\ell]$, $T_i$ has at least $2\lceil\frac{k}{4}\rceil\ell+p+2|V_{\rm NT}|-\lceil\frac{k}{4}\rceil\geq p$ leaves.
Thus,  $(G,V_{\rm NT},p,k,\ell)$ is a yes-instance of {\dLNTST}. This concludes the safety proof.

\medskip
Now, we use the observation that the instance $(G,V_{\rm NT},0,1,1)$ of {\dLNTST} is equivalent to the instance $(G,V_{\rm NT})$ of {\NTST}.
Then we again  combine the fact that the instances $(G,V_{\rm NT},p,k,\ell)$ and $(G,V_{\rm NT},0,1,1)$ of {\dLNTST} are equivalent. 
Now, we invoke the kernelization algorithm from~\Cref{prop:NTST-dec-result} to the instance $(G,V_{\rm NT})$ that outputs an equivalent instance $(G',V_{\rm NT}')$ satisfying that $|V(G')|\leq 3|V_{\rm NT}'|$.
As $(G',V_{\rm NT}')$ is equivalent  to the instance $(G',V_{\rm NT}',0,1,1)$ of {\dLIST},  we conclude that  \Cref{rule:advfinalNT} is safe. 
This concludes Case~2 and the description of the kernelization algorithm for {\dLNTST}.

The correctness of {our kernelization algorithm} follows from the safeness of Reduction Rules \ref{rule:contractNT}--\ref{rule:advfinalNT}.
Our algorithm outputs a graph with the maximum number of vertices in \Cref{rule:advfinalNT}, and $|V(G)|< (4\lceil\frac{k}{4}\rceil\ell+2p+5|V_{\rm NT}|)(\ell+6)$ in this case.
Therefore, our kernelization algorithm produces an instance of {\dLNTST} with 
$\OO((|V_{\rm NT}|+p+k\ell)\ell)$ vertices as required by the claim of the theorem.
The overall running time is polynomial because each reduction rule can be applied in polynomial time.
This completes the proof.
\end{proof}

Observe that putting $p = 0$ gives us {\DivNTST}---the diverse version of \NTST.
Hence, we have the following corollary.

\begin{corollary}
\label{cor:DNTST-result}
{\DivNTST} when parameterized by $|V_{\rm NT}| + k + \ell$ admits a kernel with $\OO((|V_{\rm NT}| + k\ell)\ell)$ vertices.
\end{corollary}

\section{Conclusions and Future Research}
\label{sec:conclusion}

Our results in this paper provide a comprehensive framework on polynomial kernelizations for the variants of the {\SpanTree} problem with constraints on the number of leaves and internal vertices in the diversity paradigm.
It would be interesting to consider other variants of the {\SpanTree} problem.
First, we would like to point out that, while disjoint spanning trees~\cite{NashWilliams64,Edmonds65} or spanning trees with the minimum sum of all pairwise Hamming distances~\cite{HanakaKKO21} can be found in polynomial time, the computational complexity of finding pairwise $k$-diverse spanning trees is open. If $k=2$, then pairwise $k$-diverse spanning trees are distinct trees. Because distinct spanning trees of a given graph can be counted by Kirchhoff's matrix tree theorem (see, e.g., \cite{Bollobas02}), it can be decided in polynomial time whether a graph has $\ell$ pairwise $2$-diverse spanning trees.  However, for $k>2$, it is unknown whether finding three pairwise $k$-diverse spanning trees is polynomial or NP-hard. The other set of open problems stems from directed graphs. What is the complexity of finding diverse spanning trees (i.e., \emph{outbranchings}) in directed graphs?  Furthermore, it is interesting to investigate the directed analogs of \DivLIST and \DivLNTST.
We remark that the Parameterized Complexity of the \textsc{Directed Max-Leaf Spanning Tree} and \textsc{Directed Max-Internal Spanning Tree} problems were considered in~\cite{AlonFGKS09,BonsmaD11,DaligaultGKY08,DaligaultT09,KneisLR11} and~\cite{CohenFGKSY10,DankelmannGK09,FominGLS12,GutinRK09}, respectively.





\end{document}